\documentclass[12pt]{article}

\usepackage{amsmath,amsthm,amsfonts,amssymb,amscd,caption,color,subcaption,cite}
\usepackage{graphicx}
\usepackage[figurename=Fig.]{caption}
\usepackage{ulem}
\numberwithin{equation}{section}
\allowdisplaybreaks[4]

\newtheorem {Lemma}{Lemma}[section]
\newtheorem {Theorem} {Theorem}[section]

\newtheorem {Proposition} {Proposition}[section]

\allowdisplaybreaks [4]
\headsep =
0.0in \headheight = 0.0in \topmargin = 0.3in

\begin{document}
\title{On  extremal leaf status and internal status of trees}

\author{Haiyan Guo\footnote{E-mail: ghaiyan0705@163.com}, Bo Zhou\footnote{Corresponding author. E-mail: zhoubo@scnu.edu.cn}
\\
School of  Mathematical Sciences, South China Normal University  \\
Guangzhou 510631, P.R. China}

\date{}
\maketitle

\begin{abstract}
For a vertex $u$ of a tree $T$, the leaf (internal, respectively) status  of $u$ is the sum of the distances from $u$ to all leaves (internal vertices, respectively) of $T$. The minimum (maximum, respectively) leaf status of a tree $T$ is the minimum (maximum, respectively) leaf statuses of all vertices of $T$. The minimum (maximum, respectively) internal status of a tree $T$ is the minimum (maximum, respectively) internal statuses of all vertices of $T$.
We give  the smallest and largest values for the minimum leaf status, maximum leaf status, minimum internal status, and maximum internal status of a tree and characterize the extremal cases. We also discuss these parameters of a tree with given  diameter or maximum degree.
\\ \\
{\bf Mathematics Subject Classifications:} 05C12,  05C35\\ \\
{\bf Keywords and phrases:} minimum leaf status, maximum leaf status, minimum internal status, maximum internal status, tree, diameter, maximum degree
\end{abstract}

\section{Introduction}

Let $G$ be a connected graph of order $n\ge 2$ with vertex set $V(G)$. For $u,v\in V(G)$, the distance between  $u$ and $v$ in $G$, denoted by $d_G(u,v)$, is the length of a shortest path connecting $u$ and $v$ in $G$. Let $\emptyset\ne A\subseteq V(G)$. For $u\in V(G)$,
the $A$-status of $u$ in $G$ is defined as
\[
s_G(u,A)=\sum_{v\in A}d_G(u,v).
\]
The minimum $A$-status of $G$ is $s_A(G)=\min\{s_G(u, A): u\in V(G)\}$, while the maximum $A$-status of $G$ is $S_A(G)=\max\{s_G(u, A): u\in V(G)\}$.
The  $A$-centroid (or $A$-median) of $G$ is defined as $\{u\in V(G): s_G(u,A)=s_A(G)\}$.

Let $s_G(u)=s_G(u,V(G))$ for $u\in V(G)$, $s(G)=s_{V(G)}(G)$ and $S(G)=S_{V(G)}(G)$. Then $s_G(u)$ is the status (or transmission) of $u$ in $G$ \cite{BH,DS,VC}, $s(G)$ is the minimum status of $G$, and $S(G)$ is the maximum status of $G$. Both minimum and maximum  statuses have been studied extensively, and it should be noted that the  minimum (maximum, respectively) status appeared also  in its normalized form divided by $n-1$ that is called the proximity (remoteness, respectively) of the graph, see, e.g., \cite{AH,AH3,AH4,Dan,Dan2,DWS,LZG,LTSZ,KA,PZ,RB, Se,Z}.


Let $T$ be a tree. For $u\in V(T)$, denote by $N_T(u)$ the set of vertices adjacent to $u$ in $T$ and the cardinality of $N_T(u)$ is the degree of $u$ in $T$, denoted by $\delta_T(u)$.
A vertex of degree one in a tree is called a leaf and a vertex of degree at least two in a tree is called an internal vertex.
Let $L(T)$ and $I(T)$ be the set of leaves and the set of internal vertices of $T$, respectively.

Slater \cite{Sla} studied structure of the $A$-centroid of a tree $T$ with $\emptyset \ne A\subseteq V(T)$.
For example, it was shown in \cite[Theorem 5]{Sla} that the $A$-centroid induces a path in a tree for any subset $A$. The another related concept is called $A$-center, which is defined to be
the set $\{u\in V(T): e_A(u,T)\}$, where $e_A(u,T)=\max\{d_T(u,v): v\in A\}$. It was shown in \cite{Sla} that $L(T)$-center and $L(T)$-centroid have quite different properties.
Wang \cite{W} characterized the trees with maximum distance between the $L(T)$-centroid and the centroid
or the  $I(T)$-centroid, and maximum distance between  the  $I(T)$-centroid and the centroid, respectively. Here for two subset $A$ and $B$ of vertices of a connected graph $G$, the distance between $A$ and $B$ is smallest distance between a vertex from $A$ and a vertex from  $B$ in $G$.

The minimum  leaf status (internal status, respectively) of $T$ is defined to be the minimum $L(T)$-status ($I(T)$-status, respectively), denoted by $ls(T)$ ($is(T)$, respectively). That is,
\[
ls(T)=s_{L(T)}(T)  \mbox{ and } is(T)=s_{I(T)}(T).
\]
The maximum leaf status (internal status, respectively) of $T$ is defined to be the maximum $L(T)$-status ($I(T)$-status, respectively), denoted by $LS(T)$ ($IS(T)$, respectively). That is,
\[
LS(T)=S_{L(T)}(T)  \mbox{ and }  IS(T)=S_{I(T)}(T).
\]
In this paper, we study the extremal properties of these four parameters of trees.
We give  the smallest and largest values for the minimum leaf status, maximum leaf status, minimum internal status, and maximum internal status of a tree and characterize the extremal cases. We also discuss these parameters of trees with given  diameter or maximum degree.

We note related work
of  Dimitrov et al. \cite{DIS}, where, if restricted to trees, they studied the extremal properties
of $\sum_{u,v\in I(T)}d_T(u,v)$
and $\sum_{u,v\in V(T)\atop\{u,v\}\cap L(T)\ne \emptyset}d_T(u,v)$ for trees $T$.

\section{Preliminaries}

The diameter of a connected graph $G$ is the maximum distance between two vertices. Denote by $S_n$ and $P_n$ the star and the path of order $n$, respectively. A  double star is a tree with diameter $3$, which is obtainable by adding an edge between the centers of two nontrivial stars.

For a vertex $u$ of a nontrivial tree $T$, the components of $T-u$ are called the branches of $T$ at $u$. For $A\subseteq V(T)$, the $A$-branch-weight of $u$ in $T$, denoted by $bw_T(u,A)$, is defined to be
\[
\max\{|A\cap V(B)|: \mbox{B is a branch of $T$ at $u$}\}.
\]

For a tree $T$, a vertex in the $A$-centroid is called an $A$-centroid vertex.
The following lemma is a restatement of \cite[Theorem 8]{Sla}.

\begin{Lemma} \label{oldLem} Let $T$ be a tree of order $n\ge 2$. Then  $u$ is an $A$-centroid vertex if and only if  $bw_T(u,A)\le bw_T(v,A)$ for any $v\in V(T)$.
\end{Lemma}

For $u,v\in V(T)$, denote by $n_T(u,v|A)$ the number of vertices in $A$ closer to $u$ than to $v$.  Let $T$ be a tree with $u\in V(T)$. For $A=L(T), I(T)$,
 Wang \cite[Proposition 3.1]{W} stated that
$u$ is an $A$-centroid vertex of $T$ if and only if $n_T(u,v|A)\ge n_T(v,u|A)$ for any $v\in N_T(u)$.

We give a somewhat easy  necessary and sufficient condition for a vertex of a tree $T$ being an  $A$-centroid vertex for $A=L(T), I(T)$.

\begin{Lemma}\label{newlemma}
Let $T$ be a tree of order $n\ge 3$ with $u\in V(T)$. For $A=L(T), I(T)$,
$u$ is an $A$-centroid vertex of $T$ if and only if $bw_T(u,A)\le \frac{|A|}{2}$.
\end{Lemma}

\begin{proof}  Let $r=\delta_T(u)$ and $N_T(u)=\{u_1,\dots,u_r\}$. For $i=1, \dots, r$, let $B_i$ be the branch of $T$  at $u$  containing $u_i$ and  let $a_i=|A\cap V(B_i)|$. Assume that $a_1\ge \dots\ge a_r$. Then, by definition, $bw_T(u,A)=a_1$.

Suppose that $bw_T(u,A)\le \frac{|A|}{2}$, i.e., $a_1\le \frac{|A|}{2}$. For any $v\in V(T)\setminus (V(B_1)\cup\{u\})$, say $v\in V(B_i)$ with  $2\le i\le r$,  as $T-V(B_i)$ is a subtree of a branch at $v$, we have $bw_T(v,A)\ge|A\cap (V(T)\setminus V(B_i))|\ge \sum_{j=1}^ra_j -a_i\ge a_1=bw_T(u,A)$.
If $A=L(T)$, then $\sum_{j=1}^ra_j=|A|$, so $\sum_{j=2}^ra_j\ge \frac{|A|}{2}$, and  for any $v\in V(B_1)$, we have $bw_T(v,A)\ge|A\cap (V(T)\setminus V(B_1))|= \sum_{j=2}^ra_j\ge \frac{|A|}{2}\ge  bw_T(u,A)$.
If  $A=I(T)$, then  $\sum_{j=1}^ra_j=|A|-1$, so $\sum_{j=2}^ra_j\ge \frac{|A|}{2}-1$, and for any $v\in V(B_1)$, we have $bw_T(v,A)\ge|A\cap (V(T)\setminus V(B_1))|= 1+\sum_{j=2}^ra_j\ge \frac{|A|}{2}\ge  bw_T(u,A)$.
Therefore $bw_T(u,A)\le bw_T(v,A)$ for any $v\in V(T)$, which implies that  $u$ is an  $A$-centroid vertex of $T$ by Lemma~\ref{oldLem}.

Conversely, suppose  that  $u$ is an $A$-centroid  vertex of $T$.

\noindent {\bf Case 1.}  $A=L(T)$.

If $|A|=2$, then $T\cong P_n$ and $u$ may be any internal vertex. Then $bw_T(u,A)=1\le \frac{|A|}{2}$.
Suppose that $|A|\ge 3$. If $bw_T(u,A)>\frac{|A|}{2}$, i.e., $a_1>\frac{|A|}{2}$, then $\sum_{i=2}^{r}a_i<\frac{|A|}{2}$,
so $s_T(u,L(T))-s_T(u_1,L(T))=a_1-\sum_{i=2}^{r}a_i>0$, implying that $s_T(u,L(T))>s_T(u_1,L(T))$,
a contradiction. It follows that $bw_T(u,A)\le \frac{|A|}{2}$.

\noindent {\bf Case 2.}  $A=I(T)$.

If $|A|=1$, then $T\cong S_n$ and $u$ is the center. If $|A|=2$, then $T$ is a double star and $u$ may be either internal vertices. So we have $bw_{T}(u,A)=0,1 \le \frac{|A|}{2}$ if $|A|=1,2$.
Suppose that $|A|\ge 3$. If $bw_T(u,A)>\frac{|A|}{2}$, i.e., $a_1>\frac{|A|}{2}$, then $\sum_{i=2}^{r}a_i<\frac{|A|}{2}-1$, so
$s_T(u,I(T))-s_T(u_1,I(T))=a_1-1-\sum_{i=2}^{r}a_i>0$, a contradiction.
It follows that $bw_T(u,A)\le \frac{|A|}{2}$.

By combining Cases 1 and 2, we have $bw_T(u,A)\le \frac{|A|}{2}$.
\end{proof}

A leaf peripherian vertex of a tree $T$ on $n$ vertices is a vertex of $T$ with maximum leaf status. Note that every vertex of $P_n$ is a leaf peripherian vertex and $LS(P_n)=n-1$. In the following lemma, we show that a leaf peripherian vertex of a tree that is not a path  must be a leaf.

\begin{Lemma} \label{peripherian}
Let $T$ be a tree that is not a path. Let $u\in V(T)$. If $u$ is a leaf peripherian vertex of $T$, then $u\in L(T)$.
\end{Lemma}

\begin{proof}
We prove the lemma by contradiction. Suppose that $u$ is a leaf peripherian vertex of $T$ but $u\notin L(T)$. Then $\delta_{T}(u)\ge 2$. Let $r=d_{T}(u)$ and $N_T(u)=\{u_1,\dots, u_r\}$, where $r\ge 2$. For $i = 1,\dots,r$, let $B_i$ be the branch of $T$ at $u$ containing $u_i$, $L_i=L(T)\cap V(B_i)$ and $a_i=|L_i|$. Assume that $a_1\le  \dots \le a_r$.

Suppose first that $r=2$ and $a_2=a_1$. As $T$ is not a path, we have $a_2=a_1\ge 2$. Let $v\in L_1$ and let $z$ be the unique vertex adjacent to $v$ in $T$. Then
\begin{eqnarray*}
s_T(u,L(T))&=&\sum_{w\in L_1\setminus\{v\}}d_T(u,w)+d_T(u,v)+\sum_{w\in L_2}d_T(u,w),\\
s_T(v,L(T))&=&\sum_{w\in L_1\setminus\{v\}}d_T(v,w)+\sum_{w\in L_2}(d_T(v,u)+d_T(u,w)),
\end{eqnarray*}
and so
\begin{eqnarray*}
&&s_T(v,L(T))-s_T(u,L(T))\\
&=&\sum_{w\in L_1\setminus\{v\}}(d_T(v,w)-d_T(u,w))-d_T(u,v)+\sum_{w\in L_2}d_T(v,u)\\
&>&\sum_{w\in L_1\setminus\{v\}}(d_T(z,w)-d_T(u,w))-d_T(u,v)+\sum_{w\in L_2}d_T(v,u)\\
&\ge&-\sum_{w\in L_1\setminus\{v\}}d_T(z,u)-d_T(u,v)+\sum_{w\in L_2}d_T(v,u)\\
&=&-\sum_{w\in L_1\setminus\{v\}}(d_T(v,u)-1)-d_T(u,v)+\sum_{w\in L_2}d_T(v,u)\\
&=& \sum_{w\in L_1\setminus\{v\}}1\\
&=&a_1-1\\
&>&0.
\end{eqnarray*}
Thus $s_T(v,L(T))>s_T(u,L(T))$. This implies that $u$ can not  be a leaf peripherian vertex of $T$, a contradiction.

Suppose next that $r=2$ and $a_2>a_1$, or $r\ge 3$.
Then
\begin{eqnarray*}
s_T(u,L(T))&=&\sum_{v\in L_1}d_T(u,v)+\sum_{j=2}^{r}\sum_{v\in L_j}d_T(u,v),\\
s_T(u_1,L(T))&=&\sum_{v\in L_1}(d_T(u,v)-1)+\sum_{j=2}^{r}\sum_{v\in L_j}(d_T(u,v)+1),
\end{eqnarray*}
and so
\begin{eqnarray*}
s_T(u_1,L(T))-s_T(u,L(T))&=&\sum_{v\in L_1}(-1)+\sum_{j=2}^{r}\sum_{v\in L_j}1\\
&=&\sum_{i=2}^{r}a_i-a_1\\
&>&0.
\end{eqnarray*}
Thus $s_T(u_1,L(T))>s_T(u,L(T))$. This implies that $u$ is  not a leaf peripherian vertex of $T$, also a contradiction.

Therefore, $u\in L(T)$, as desired.
\end{proof}

A internal peripherian vertex of a tree $T$ is a vertex of $T$ with maximum internal status.

\begin{Lemma} \label{inperipherian}
Let $T$ be a tree. Suppose that $u$ is a internal peripherian vertex of $T$. Then $u\in L(T)$.
\end{Lemma}

\begin{proof}
We prove the lemma by contradiction. Suppose that $u\notin L(T)$. Then $\delta_{T}(u)\ge 2$.

If there is a vertex $v\in N_T(u)$ with $\delta_T(v)=1$, then it is obvious that $s_T(v,I(T))>s_T(u,I(T))$, a contradiction. So $\delta_T(v)\ge 2$ for any $v\in N_T(u)$. Let $r=\delta_{T}(u)$ and $N_T(u)=\{u_1,\dots, u_r\}$, where $r\ge 2$. For $i = 1,\dots,r$, let $B_i$ be the branch of $T$ at $u$ containing $u_i$, $I_i=I(T)\cap V(B_i)$ and $a_i=|I_i|$. Assume  that $a_1\le  \dots \le a_r$. Let $z\in I_1$. Then
\begin{eqnarray*}
s_T(u,I(T))&=&\sum_{w\in I_1\setminus\{z\}}d_T(u,w)+d_T(u,z)+\sum_{i=2}^{r}\sum_{w\in I_i}d_T(u,w),\\
s_T(z,I(T))&=&\sum_{w\in I_1\setminus\{z\}}d_T(z,w)+d_T(z,u)+\sum_{i=2}^{r}\sum_{w\in I_i}(d_T(z,u)+d_T(u,w)),
\end{eqnarray*}
and so
\begin{eqnarray*}
&&s_T(z,I(T))-s_T(u,I(T))\\
&=&\sum_{w\in I_1\setminus\{v\}}(d_T(z,w)-d_T(u,w))+\sum_{i=2}^{r}\sum_{w\in I_i}d_T(z,u)\\
&\ge&-d_T(z,u)(a_1-1)+d_T(z,u)\sum_{i=2}^{r}a_i\\
&=&d_T(z,u)\left(-a_1+1+\sum_{i=2}^{r}a_i\right)\\
&>&0.
\end{eqnarray*}
Thus $s_T(z,I(T))>s_T(u,I(T))$, a contradiction.
\end{proof}

A tree is called starlike if it has at most one vertex of degree greater than $2$. So, a star and a path are both particular starlike trees.

A diametric path of a tree is a longest path in this tree (whose length equals the diameter). Evidently, the terminal vertices of a diametric path of any nontrivial tree are leaves.

A caterpillar is a tree such that the deletion of all leaves outside a diametric path (if any exists) yields a path. 

A leaf edge in a tree is an edge incident with a leaf.

For a tree $T$ with $uw\in E(T)$ and $vw\not\in E(T)$, if $T'=T-uw+vw$ is a tree, then we also say that
 $T'$ is obtained from $T$ by moving the edge $uw$ from $u$ to $v$.

A hanging path at a vertex $u$ of a tree $T$ is a path $uu_1\dots u_{\ell}$ with $\delta_T(u)\ge 3$, $\delta_T(u_{\ell})=1$ and if $\ell\ge 2$, $\delta_T(u_i)=2$ for $i=1, \dots, \ell-1$.

Let $P$ be a path in a tree $T$. For $v\in V(T)\setminus V(P)$, the distance between $v$ and $P$ is defined to be $d_T(v,P)=\min\{d_T(v, w): w\in V(P)\}$.

\section{Minimum leaf status}

\begin{Theorem} \label{lfmin}
Let $T$ be a tree of order $n$. Then
\[
ls(T)\ge n-1
\]
with equality if and only if $T$ is  starlike.
\end{Theorem}

\begin{proof}
Let $T$ be a tree of order $n$ that minimizes the minimum leaf status. Let $u$ be an $L(T)$-centroid vertex. Suppose that there is a vertex $v$, different from $u$, with degree at least $3$.  Denote by $v_0, v_1, \dots, v_{r-1}$ all the neighbors of $v$, where $v_0$ lies on the unique path connecting $u$ and $v$ in $T$, and $r=\delta_T(v)$.
Let
\[
T'=T-\{vv_i: i=1, \dots, r-2\}+\{uv_i: i=1, \dots, r-2\}.
\]
Let $V'$ be the set of leaves in the branch of $T$  at $u$ containing $v$, and $V''$ the set of leaves of all the branches of $T$ at $v$ containing one of $v_1, v_2,\dots, v_{r-2}$. Then
\begin{eqnarray*}
s_T(u, L(T))-s_{T'}(u, L(T')) &=& s_T(u, V')-s_{T'}(u, V')\\
&=& \sum_{w\in V''}(d_T(u,w)-d_{T'}(u,w))\\
&=& \sum_{w\in V''}d_T(u,v)\\
&=& |V''|d_T(u,v)\\
&>& 0,
\end{eqnarray*}
and so $ls(T)=s_T(u, L(T))>s_{T'}(u, L(T'))\ge ls(T')$, a contradiction. Thus, all vertices different from $u$ are of degree $1$ or $2$. That is, there is at most one vertex of degree at least $3$, or $T$ is starlike. By Lemma \ref{newlemma}, in a starlike tree $T$ of order $n$, the vertex of maximum degree is an $L(T)$-centroid vertex, and thus $ls(T)=|E(T)|=n-1$.
\end{proof}

For integers $n$, $a$ and $b$ with $2\le a, b\le \frac{n-2}{2}$, $T_{n;a,b}$ be the tree of order $n$ obtained from two stars $S_{a+1}$ and  $S_{b+1}$ by connecting their centers by a path of length $n-a-b-1$. For convenience, let $T_{n,a}=T_{n;a,a}$.

\begin{Theorem}
Let $T$ be a tree of order $n\ge 6$. Then
\[
ls(T)\le \left\lfloor\frac{(n+1)^2}{8}\right\rfloor
\]
with equality if and only if $T\cong T_{n, \lceil\frac{n}{4}\rceil}$  if $n$ is even or  $n\equiv 3~(\bmod~4)$, and $T\cong T_{n, \frac{n-1}{4}}, T_{n,\frac{n+3}{4}}$ if $n\equiv 1~(\bmod~4)$.
\end{Theorem}

\begin{proof}
Let $T$ be a tree of order $n$ that maximizes the minimum leaf status.

Let $x$ be an $L(T)$-centroid vertex.
Let $r=\delta_T(x)$ and $N_T(x)=\{x_1, \dots, x_r\}$. For $i=1, \dots, r$, let $B_i$ be the branch of $T$ at $x$ containing $x_i$ and  let $a_i=|L(T)\cap V(B_i)|$. Assume that $a_1\ge \dots\ge a_r$.

By Theorem~\ref{lfmin}, $T$ can not be a starlike tree and thus there are at least two vertices of degree at least three, and it is obvious that  one such vertex lies in some branch  of $T$ at $x$.

\noindent {\bf Claim 1.} If a branch  of $T$ at $x$ contains a vertex of degree at least three in $T$, then there is exactly one such vertex in this branch, and its neighbors in $T$ are all leaves except one lying  on the unique path connecting to $x$.

Assume that $B_1$ is a branch of $T$  at $x$ with a vertex of degree at least three in $T$.
Let $P=v_0v_1\dots v_p$ be a longest path in $T$ from a leaf of $T$ in $B_1$ to $x$, where $v_p=x$, $v_{p-1}=x_1$ and $p\ge 2$. Then $\delta_T(v_i)\ge 3$ for some $i=1, \dots, p-1$. Suppose that  $i>1$. Then  $p\ge 3$.
Let $T'$ be the tree obtained from $T$ by moving all edges outside $P$ incident with $v_i$ from $v_i$ to $v_{1}$. By Lemma~\ref{newlemma}, $x$ is an  $L(T')$-centroid vertex. It is evident that $ls(T')>ls(T)$, a contradiction. Thus,
$v_{1}$ is the only vertex in $B_1$ with degree at least three in $T$  and all its  neighbors are leaves of $T$  except  the one lying on the path connecting to $x$. This proves Claim 1.

\noindent {\bf Claim 2.}   $T$ is obtainable by connecting the centers of two copies of $S_{a_1+1}$ by a path, where $a_1\ge 2$.

Suppose first that $r=2$. Then, by Lemma~\ref{newlemma},  both $B_1$ and $B_2$ contain a vertex of degree at least three in $T$, and by Claim 1, there is exactly one vertex of degree at least three in each of $B_1$ and $B_2$, and its neighbors in $T$ are all leaves except one lying  on the path connecting to $x$. So, $T$ is obtainable by connecting the centers of two stars, say  $S_{a_1+1}$ and
$S_{a_2+1}$, by a path, where $a_1\ge a_2\ge 2$. Suppose that $a_1>a_2$. Let $z$ be the vertex in $B_1$ with degree $a_1+1$ in $T$. Then $bw_T(x,L(T))=a_1>a_2=bw_T(z,L(T))$. So $x$ is not an $L(T)$-centroid vertex by Lemma~\ref{oldLem}, a contradiction. Thus  $a_1=a_2$.

Suppose next that $r\ge 3$.
Assume that $B_1$ is a branch of $T$ at $x$ containing a vertex of degree at least three in $T$. By Claim 1, there is exactly one such vertex $v_1$, and the neighbors of $v_1$ in $T$ are all leaves except the one lying  on the path between $v_1$ and $x$. Let $z\in L(T)\cap V(B_r)$, $k=d_T(x,z)$ and $\ell=d_T(x, v_1)$.
By Lemma~\ref{oldLem}, $a_1=bw_T(x,L(T))\le bw_T(v_1,L(T))=\sum_{i=2}^{r}a_i$.

Suppose that $a_1<\sum_{i=2}^{r}a_i$.
Suppose first that  $a_r=1$.
Let
\[
T'=T-xx_1+x_1z.
\]
Note that $\sum_{i=1}^{r}a_i=|L(T)|$. As $a_1<\sum_{i=2}^{r}a_i$, we have $2a_1<|L(T)|$, so $a_1\le \frac{|L(T)|-1}{2}=\frac{|L(T')|}{2}$. Then, by Lemma~\ref{newlemma}, $x$ is an $L(T')$-centroid vertex. Thus
\begin{eqnarray*}
ls(T')-ls(T)&=&s_{T'}(x, L(T'))-s_{T}(x, L(T))\\
&=&(k+\ell+1)a_1-k-(\ell+1)a_1\\
&=& (a_1-1)k\\
&>&0.
\end{eqnarray*}
That is, $ls(T')>ls(T)$, a contradiction.
Suppose next that  $a_r\ge 2$. Let $z_1, \dots, z_{a_r}$ be the leaves of $T$ in $B_r$, where $z_1=z$. By Claim 1, these leaves are adjacent to a common vertex, say $v_r$.
We consider $a_1\le\sum_{i=2}^{r-1}a_i$ and $a_1>\sum_{i=2}^{r-1}a_i$ separately.
In the former case, let
\[
T'=T-xx_1-\{v_rz_i: i=2, \dots, a_r\}+\{z_iz_{i+1}: i=1, \dots, a_r-1\}+x_1z_{a_r}.
\]
As $a_1\le \frac{|L(T)|-a_r}{2}=\frac{|L(T)'|}{2}$, $x$ is an $L(T')$-centroid vertex by Lemma~\ref{newlemma}, so
\[
ls(T')-ls(T)=(k+a_r-1+\ell+1)a_1-ka_r-(\ell+1)a_1=(k+a_r-1)a_1-ka_{r}>0,
\]
i.e.,   $ls(T')>ls(T)$,  a contradiction.
In the latter case,
let
\[
T'=T-\{xx_i:i=2,\dots,r-1\}+\{x_rx_i:i=2,\dots,r-1\}.
 \]
By Lemma~\ref{newlemma}, $x_r$ is an $L(T')$-centroid vertex. Then $ls(T')-ls(T)=a_1-a_r>0$, also a contradiction.
Therefore,  $a_1=\sum_{i=2}^{r}a_i$.

If $|V(B_i)|>1$ for some $2\le i\le r$, then for the tree
\[
T''=T-\{xx_{j}: j=2,\dots, r \mbox{ with }  j\ne i\}+\{x_ix_{j}: j=2,\dots, r \mbox{ with }   j\ne i\},
\]
as $x$ is an $L(T'')$-centroid vertex by Lemma~\ref{newlemma}, we have  $ls(T'')>ls(T)$. This contradiction shows that $|V(B_2)|=\cdots=|V(B_r)|=1$. This proves Claim 2.

By Claim 2, $T\cong T_{n, a_1}$. Then the diameter of $T$ is $d=n+1-2a_1$, and any  internal vertex is an $L(T)$-centroid vertex. Let $v$ be the neighbor of some leaf.   So
\[
ls(T)=s_T(v,L(T))=a_1+a_1(d-1)=a_1d=a_1(n+1-2a_1):=f(a_1),
\]
which is strictly increasing when $a_1\le \frac{n+1}{4}$ and strictly decreasing when $a_1\ge \frac{n+1}{4}$.
If $n\equiv 3~(\bmod~4)$, then $f(a_1)\le f\left(\frac{n+1}{4}\right)=\frac{(n+1)^2}{8}=\left\lfloor\frac{(n+1)^2}{8}\right\rfloor$ with equality if and only if $a_1=\frac{n+1}{4}$.
If  $n\equiv 2~(\bmod~4)$, then  $f(a_1)\le f\left(\frac{n+2}{4}\right)=\frac{n(n+2)}{8}=\left\lfloor\frac{(n+1)^2}{8}\right\rfloor$ with equality if and only if $a_1=\frac{n+2}{4}$.
If $n\equiv 0~(\bmod~4)$, then $f(a_1)\le f\left(\frac{n}{4}\right)=\frac{n(n+2)}{8}=\left\lfloor\frac{(n+1)^2}{8}\right\rfloor$ with equality if and only if $a_1=\frac{n}{4}$.
If $n\equiv 1~(\bmod~4)$, then $f(a_1)\le f\left(\frac{n-1}{4}\right)=f\left(\frac{n+3}{4}\right)=\frac{(n-1)(n+3)}{8}=\left\lfloor\frac{(n+1)^2}{8}\right\rfloor$
 with equality if and only if $a_1=\frac{n-1}{4}$ or $\frac{n+3}{4}$. So,
we have $T\cong T_{n, a}$, where
 \[
a=\begin{cases} \lceil\frac{n}{4}\rceil  & \mbox{if $n$ is even or  $n\equiv 3~(\bmod~4)$}\\[2mm]
\frac{n-1}{4}, \frac{n+3}{4} & \mbox{if $n\equiv 1~(\bmod~4$})
\end{cases}
\]
and $ls(T)=\left\lfloor\frac{(n+1)^2}{8}\right\rfloor$.
\end{proof}

\begin{Theorem} Let $T$ be a tree of order $n$ with diameter $d$, where $3\le d\le n-1$. Then
\[
ls(T)\le\begin{cases} \frac{(n-d+1)d}{2}  & \mbox{if $n-d$ is odd}\\[2mm]
\frac{(n-d)d}{2}+1 & \mbox{if $n-d$ is even}
\end{cases}
\]
with equality if and only if $T\cong T_{n,\frac{n-d+1}{2}}$ when $n-d$ is odd,  and $T\cong T_{n;\frac{n-d}{2},\frac{n-d+2}{2}}$, or $d\ge 4$ and $T$ is isomorphic to a tree
obtained from
$T_{n-1,\frac{n-d}{2}}$ by adding a leaf edge at a vertex of degree two when $n-d$ is even.
\end{Theorem}

\begin{proof} If $d=n-1$, then the result is trivial as $n-d=1$ is odd,  $T\cong P_n$ and $ls(T)=n-1=\frac{(n-d+1)d}{2}$.

Suppose that $d\le n-2$.
Let $T$ be a tree of order $n$  with diameter $d$ that maximizes the minimum leaf status. Let $x$ be an $L(T)$-centroid vertex. Let $r=\delta_T(x)$ and $N_T(x)=\{y_1, \dots, y_{r}\}$. For $i=1, \dots, r$, let $B_i$ be the branch of $T$ at $x$ containing $y_i$ and let $a_i=|L(T)\cap V(B_i)|$.

Suppose that $r=2$.  By Lemma~\ref{newlemma}, $a_1, a_2\le\frac{|L(T)|}{2}$. So $a_1= a_2=\frac{|L(T)|}{2}$. Choose a vertex $w\in V(B_{1})$ with $\delta_T(w)\ge 3$ such that $d_T(x,w)$ is as small as possible. Then $bw_T(w,L(T))=a_2=\frac{|L(T)|}{2}$. So $w$ is an $L(T)$-centroid vertex by Lemma~\ref{newlemma}. Therefore, we may assume that $r\ge 3$. Let $P:=x_0x_1\dots x_d$ be an arbitrary diametric path of $T$.

\noindent {\bf Claim 1.} $x$ is in some diametric path.

Suppose this is not true. That is, $x$ lies outside any diametric path.  Then, for some $i$ with $3\le i\le d-3$, one has
$d_T(x, P)=d_T(x,x_i)\le\min\{i-2,d-i-2\}$. Assume that $x_i\in B_r$. Suppose that  $a_j>a_r$ for some $j$ with $1\le j\le r-1$.
Let $T'$ be the tree obtained from $T$ by moving the edges $xy_k$ with $1\le k\le r-1$ and $k\neq j$
from $x$ to $y_r$. Note that the diameter of $T'$ is $d$ and $L(T')=L(T)$. As $a_j>a_r$, we have $bw_{T'}(y_r)=\max\{a_1, \dots, a_{r-1}\}=bw_{T}(x)$. 
By Lemma~\ref{newlemma},  $y_r$ is an $L(T')$-centroid vertex. So $ls(T)=s_T(x,L(T))<s_T(x,L(T))+a_j-a_r=s_{T'}(y_r,L(T'))=ls(T')$. This contradiction shows that $a_j\le a_r$ for $1\le j\le r-1$.

Suppose that, for some $1\le j\le r-1$, $a_j<a_r$, and either $\delta_T(y_j)\ge 2$ or $\delta_T(y_j)=1$ and $a_r<\frac{|L(T)|}{2}$.

Let $T''$ be the tree obtained from $T$ by moving the edges $xy_k$ with $1\le k\le r-1$ and $k\ne j$ from $x$ to $y_j$.  Evidently,  the diameter of $T''$ is $d$. Let $T_0$ be the maximal subtree of $T-xy_j$ containing $y_j$.
Note that the branches of $T''-y_j$ are  $B_k$ with $1\le k\le r-1$ and $k\ne j$, $T[\{x\}\cup V(B_r)]$, and if $\delta_T(y_j)\ge 2$, the branches of $T_0$ at $y_j$.
As  $a_k\le a_r$ for $1\le k\le r-1$,
$bw_{T''}(y_j)=\max\{a_k: k=1, \dots, r, k\ne j\}
=a_r$.

If $\delta_T(y_j)\ge 2$, then
 $bw_{T''}(y_j)
=a_r=bw_{T}(x)\le\frac{|L(T)|}{2}=\frac{|L(T'')|}{2}$ by Lemma~\ref{newlemma}.
If $\delta_T(y_j)=1$ and $a_r<\frac{|L(T)|}{2}$, then
$bw_{T''}(y_j)
=a_r\le\frac{|L(T)|-1}{2}=\frac{|L(T'')|}{2}$. In either case,   $y_j$ is an $L(T'')$-centroid vertex by Lemma~\ref{newlemma}. So $ls(T)=s_T(x,L(T))<s_T(x,L(T))+a_r-a_j=s_{T''}(y_j,L(T''))=ls(T'')$, a contradiction. Therefore $a_1=\dots=a_r$, or $\delta_T(y_j)=1$ for $1\le j\le r-1$ and $a_r=\frac{|L(T)|}{2}$. We show both cases are impossible.

\noindent {\bf Case 1.} $\delta_T(y_j)=1$ for $1\le j\le r-1$ and $a_r=\frac{|L(T)|}{2}$.

Obviously, $r-1=\frac{|L(T)|}{2}$. Let $w$ be a vertex in $B_r$ with $\delta_T(w)\ge 3$ such that $d_T(w,x)$ is as small as possible. Note that $bw_T(w)=\frac{|L(T)|}{2}$. So $w$ is an $L(T)$-centroid vertex by Lemma~\ref{newlemma}. Suppose that  $w\neq x_i$. Let $w'$ be the neighbor of $w$
in the branch of $T$ at $w$ containing $x_i$. Let $T'$ be the tree obtained from $T$ by moving edges incident with $w$ outside the path connecting $x$ and $x_i$ from $w$ to $w'$. Then $bw_{T'}(w')=\frac{|L(T)|}{2}=\frac{|L(T')|}{2}$. Thus $w'$ is an $L(T')$-centroid vertex by Lemma~\ref{newlemma}. It is evident that $s_{T'}(w',L(T'))>s_{T}(w,L(T))$, implying that
$ls(T')=s_{T'}(w',L(T'))>s_{T}(w,L(T))=ls(T)$, a contradiction. It thus follows that $w=x_i$. Let $T''=T-xy_1+x_1y_1$. As $bw_{T''}(x_i)\le\frac{|L(T'')|}{2}$, $x_i$ is an $L(T'')$-centroid vertex by Lemma~\ref{newlemma}. Then $ls(T'')=s_{T''}(x_i,L(T''))>s_{T}(x_i,L(T))=s_{T}(x,L(T))=ls(T)$, a contradiction. So Case 1 can not occur.

\noindent {\bf Case 2.} $a_1=\dots=a_r$.

Let $w$ be a vertex in $B_r$ with $\delta_T(w)\ge 3$ such that $d_T(x,w)$ is as small as possible. Let $T'=T-\{xy_j:2\le j\le r-1\}+\{wy_j:2\le j\le r-1\}$. Note that $bw_{T'}(w)=bw_{T}(x)=a_1<\frac{|L(T)|}{2}=\frac{|L(T')|}{2}$. So $w$ is an $L(T')$-centroid vertex by Lemma~\ref{newlemma}.
Thus  $ls(T')-ls(T)=s_{T'}(w,L(T'))-s_{T}(x,L(T))=a_1d_T(x,w)-a_rd_T(x,w)=0$, i.e.,
$ls(T')=ls(T)$.

Suppose that $w\neq x_i$.  Let $t= \delta_{T'}(w)$ and $N_{T'}(w)=\{w_1,\dots,w_t\}$. For $j=1,\dots, t$, let $B'_j$ be the branch of $T$ at $w$ containing $w_j$ and let $b_j=|L(T')\cap V(B'_j)|$. Assume that $x\in V(B'_1)$ and $x_i\in V(B'_t)$. Let $T''=T'-\{ww_i:2\le j\le t-1\}+\{w_tw_i:2\le j\le t-1\}$. Obviously, $bw_{T''}(w_t)=b_1=a_1$, so $w_t$ is an $L(T'')$-centroid vertex by Lemma~\ref{newlemma}. Note also that $b_1=a_1=a_r>b_t$. Then $ls(T'')=s_{T''}(w_t,L(T''))=s_{T'}(w,L(T'))+b_1-b_t>s_{T'}(w,L(T'))=ls(T')=ls(T)$, a contradiction.
Thus $w=x_i$. Let $n_1$ ($n_2$, respectively) be the number of leaves in the branch of $T'$  at $x_i$ containing $x_{i-1}$ ($x_{i+1}$, respectively).
Let $z_1$ be the neighbor of $x_i$ in the path connecting $x_i$ and $x$ in $T'$.
Assume that $n_1\ge n_2$. Suppose that $n_1>n_2$.
Denote by  $T_1$ the tree obtained from $T'$ by moving the edges $x_iz_1, x_iy_2, \dots, x_iy_{r-1}$ from $x_i$ to $x_{i+1}$. By Lemma~\ref{newlemma}, $bw_{T_1}(x_{i+1})=a_1=bw_{T'}(x_{i})\le\frac{|L(T')|}{2}=\frac{|L(T_1)|}{2}$, and so $x_{i+1}$ is an $L(T_1)$-centroid vertex. Thus $ls(T_1)=s_{T_1}(x_{i+1},L(T_1))\ge s_{T'}(x_i,L(T'))+n_1-n_2>s_{T'}(x_i,L(T'))=ls(T')=ls(T)$. This contradiction shows that $n_1=n_2$.
Let $T_2=T'-\{x_iz_1, x_iy_2\}+\{x_{i-1}z_1, x_{i+1}y_2\}$. Note that there are $n_1+a_1$ leaves in the branch of $T_2$ at $x_i$ containing $x_0$ and $n_2+a_2$ ($=n_1+a_1$) leaves in the branch of $T_2$ at $x_i$ containing $x_d$. Thus $bw_{T_2}(x_i)=a_1+n_1\le \frac{|L(T_2)|}{2}$, implying that $x_i$ is an $L(T_2)$-centroid vertex by Lemma~\ref{newlemma}. It follows that $ls(T_2)=s_{T_2}(x_i,L(T_2))=s_{T'}(x_i,L(T'))+2a_1>s_{T'}(x_i,L(T'))=ls(T')$, also a contradiction.
So Case 2 can not occur.

Now Claim 1 follows by combining the above two cases.

By Claim $1$, we may assume that  $x_i=x$, where $1\le i\le d-1$. Denote by $a$ the number of leaves in the branch of $T$ at $x_i$ containing $x_{i-1}$.

\noindent {\bf Claim 2.} $T$ is a caterpillar.

Suppose that  $T$ is not a caterpillar. Then $d(w,P)\ge 2$ for some leaf $w$ of $T$ outside $P$. Let $u$ be the unique neighbor vertex of $w$. Assume that  $d(w,x_j)=d(w,P)$ and $j\le i$.   We want to show that $j=i$. Suppose that $j<i$. Choose $j$ such that $i-j$ is as small as possible. If $j<i-1$ and $\delta_T(x_k)\ge 3$ for some $k$ with $j<k<i$, then by moving the leaf edges at $x_k$ to $x_1$, we get a tree $T'$ for which $x$ is still an $L(T')$-centroid by Lemma~\ref{newlemma}, so $ls(T')>ls(T)$, which is a contradiction.
Thus  $\delta_T(x_{j+1})=\dots=\delta_T(x_{i-1})=2$ if $j<i-1$.
Denote by $T'$ the tree obtained from $T$ by moving all the leaf edges at $u$ from $u$ to $x_1$.
If $x_i$ is also an $L(T')$-centroid vertex, then, as $d_T(x_1,x_i)\ge d_T(u,x_i)$ and $u\in L(T')$, we have $ls(T')=s_{T'}(x_i,L(T'))>s_{T}(x_i,L(T))=ls(T)$, a contradiction. Thus $x_i$ is not an $L(T')$-centroid vertex. Note that the branches of $T'$ at $x_i$ containing no $x_j$ are just the branches of $T$ at $x_i$ not containing $x_j$, and that $x_i$ is an $L(T)$-centroid vertex.  By Lemma~\ref{newlemma},  $a +1=bw_{T'}(x)>\frac{|L(T')|}{2}$ and  $a \le \frac{|L(T)|}{2}$. So $\frac{|L(T)|+1}{2}=\frac{|L(T')|}{2}< a +1\le \frac{|L(T)|}{2}+1$, i.e., $a =\frac{|L(T)|}{2}$. Then $bw_T(x_j,L(T))=bw_{T'}(x_j,L(T'))=|L(T)|-a=\frac{|L(T)|}{2}$, so  $x_j$ is an $L(T)$-centroid vertex and also an $L(T')$-centroid vertex  by Lemma~\ref{newlemma}. Thus  $ls(T')=s_{T'}(x_j,L(T'))>s_{T}(x_j,L(T))=ls(T)$, also a contradiction. This shows that  $j=i$.

By the choice of $T$, each internal vertex of $T$ on $P$  different from $x_1,x_{d-1},x_i$ has degree two. Otherwise, as above, by moving the leaf edges from  these vertices to  $x_1$ or $x_{d-1}$ would result in a tree with larger minimum leaf status.
Note that the number of leaves in the branch of $T$ at $x_{i}$ containing  $x_{i-1}$ is $a$.
Let $p=\delta_T(u)-1$.  Let $k$ be the number of leaves in the branch of $T$ at $x_i$ containing  $w$.
Let $n_0$ be the maximum  number of leaves in a branch of $T$ at $x_i$ containing no $x_{i-1}$ and $w$.
Suppose that $a +p\le \frac{|L(T)|+1}{2}$. Let $T'$ be the tree obtained from $T$ by moving the leaf edges at $u$ from $u$ to $x_1$.  By Lemma~\ref{newlemma}, $bw_{T'}(x_i)=\max\{a +p,k-p+1,n_0\}\le\max\left\{a +p,\frac{|L(T)|}{2}\right\}\le\frac{|L(T)|+1}{2}$, and so $x_i$ is an $L(T')$-centroid vertex. Then we have $ls(T')=s_{T'}(x_i,L(T'))>s_{T}(x_i,L(T))=ls(T)$, a contradiction. It thus follows that  $a +p>\frac{|L(T)|+1}{2}$, i.e., $p>\frac{|L(T)|+1}{2}-a $.  We form a tree $T''$ by  moving $\left\lfloor\frac{|L(T)|+1}{2}\right\rfloor-a$ leaf edges at $u$ from $u$ to $x_1$ and the remaining leaf edges at $u$ from $u$ to $x_{d-1}$.
Note that $bw_{T''}(x_i)=\left\lfloor\frac{|L(T)|+1}{2}\right\rfloor\le \frac{|L(T'')|}{2}$. Thus $x_i$ is also an $L(T'')$-centroid vertex by Lemma~\ref{newlemma}. So $ls(T'')=s_{T_2}(x_i,L(T_2))>s_{T}(x_i,L(T))=ls(T)$, a contradiction. This completes the proof  Claim 2.

By Claim 2, $T$ is a caterpillar. Then  $|L(T)|=n-d+1$.
Note that  $a =\delta_T(x_1)-1$. Let $b=\delta_T(x_{d-1})-1$.

\noindent{\bf Case 1.}  $n-d$ is odd.

We want to show $i=1$ or $d-1$.  Suppose that this is not true.
  By the choice of $T$, each vertex from $x_2, \dots, x_{d-2}$ different from $x_i$ has degree two in $T$.
 Note that the number of leaves at $x_i$ in $T$ is $n-d+1-a -b$.
By Lemma~\ref{newlemma},  $a, b\le \frac{n-d+1}{2}$.
Suppose that  $a <\frac{n-d+1}{2}$.
We form a tree $T'$ by  moving $\frac{n-d+1}{2}-a $ leaf edges at $x_i$ from $x_i$ to $x_1$.
Evidently,  $ls(T')=s_{T'}(x_i,L(T'))>s_{T}(x_i,L(T))=ls(T)$, a contradiction. So $a =\frac{n-d+1}{2}$. Similarly, $b=\frac{n-d+1}{2}$.
 Thus $T\cong T_{n,\frac{n-d+1}{2}}$ and $ls(T)=\frac{(n-d+1)d}{2}$.

\noindent{\bf Case 2.}  $n-d$ is even.

If $i=1$ or $d-1$, then with a similar argument as in Case 1, we have $T\cong T_{n;\frac{n-d}{2},\frac{n-d+2}{2}}$ and $ls(T)=\frac{(n-d)d}{2}+1$. Suppose that $i\neq 1, d-1$.  Then $d\ge 4$. By the choice of $T$, each vertex from $x_2, \dots, x_{d-2}$ different from $x_i$ has degree two in $T$.
By Lemma~\ref{newlemma}, $a, b\le\frac{n-d}{2}$.  Suppose that $a <\frac{n-d}{2}$.
Then a tree $T'$ can be formed by moving $\frac{n-d}{2}-a $ leaf edges at $x_i$ from $x_i$ to $x_1$. 
Thus $ls(T')=s_{T'}(x_i,L(T'))>s_{T}(x_i,L(T))=ls(T)$, a contradiction. Thus $a =\frac{n-d}{2}$. Similarly, $b=\frac{n-d}{2}$. Thus  $T$ is isomorphic to a tree obtained from
$T_{n-1,\frac{n-d}{2}}$ by adding a leaf edge at a vertex of degree two,
 $ls(T)=\frac{(n-d)d}{2}+1$.

 The result follows by combining the above two cases.
\end{proof}

\section{Maximum leaf status}

\begin{Theorem}\label{LSmin}
Let $T$ be a tree of order $n$. Then
\[
LS(T)\ge n-1
\]
with equality if and only if $T\cong P_n$.
\end{Theorem}

\begin{proof} If  $T\cong P_n$, then $LS(T)=n-1$.

Suppose that  $T$ is not a path. Then $n\ge 4$. Let $x$ be a leaf peripherian vertex of $T$. By Lemma~\ref{peripherian}, $x\in L(T)$. So $LS(T)=s_T(x, L(T))=\sum_{y\in L(T)\setminus\{x\}}d_T(x,y)$. Let $u$ be a vertex of degree at least three such that $d_T(x,u)$ is as small as possible. Then the unique path from $x$ to any other leaf of $T$ contains the path from $x$ to $u$. On the other hand, every edge of $T$ lies on some path connecting $x$ and some other leaf of $T$. So
\begin{eqnarray*}
LS(T)&=&\sum_{y\in L(T)\setminus\{x\}}d_T(x,y)\\
&\ge & d_T(x,u)(|L(T)|-1)+n-1-d_T(x,u)\\
&\ge& |L(T)|-2+n-1\\
&>& n-1.
\end{eqnarray*}
So the result follows.
\end{proof}

For integers $n$, $a$ with $1\le a\le n-2$, let $P_{n,a}$ be the tree of order $n$ obtained by identifying the
center of a star $S_{a+1}$ and a terminal vertex of a path $P_{n-a}$. Particularly, $P_{n,1}=P_n$ and  $P_{n,n-2}=S_n$.

\begin{Theorem}
Let $T$ be a tree of order $n\ge 4$. Then
\[
LS(T)\le \left\lfloor\frac{n^2}{4}\right\rfloor
\]
with equality if and only if $T\cong P_{n,\frac{n}{2}}$ for even $n$, and  $T\cong P_{n,\frac{n-1}{2}}, P_{n,\frac{n+1}{2}}$ for odd $n$.
\end{Theorem}

\begin{proof}
If $n=4$, then $T\cong P_4$ or $S_4$ $(=P_{4,2})$ and $LS(P_4)=3<4=LS(P_{4,2})$. If $n=5$, then $T\cong P_5$, $P_{5,2}$ or $S_5$ $(=P_{5,3})$ and $LS(P_5)=4<6=LS(P_{5,2})=LS(P_{5,3})$. So the result holds if $n=4,5$.

Suppose that $n\ge 6$. Let $T$ be a tree of order $n$ that maximizes the maximum leaf status.
Note that the maximum leaf status of $P_{n,\frac{n}{2}}$ for even $n$, and  $P_{n,\frac{n-1}{2}}$ or  $P_{n,\frac{n+1}{2}}$ for odd $n$ is $\left\lfloor\frac{n^2}{4}\right\rfloor$.
If  $T\cong P_n$, then $LS(T)=n-1<\left\lfloor\frac{n^2}{4}\right\rfloor$.
If $T\cong S_n$, then $LS(T)=2(n-2)<\left\lfloor\frac{n^2}{4}\right\rfloor$.
So  $T$ is neither a path nor a star.

Let $x$ be a leaf peripherian vertex of $T$. By Lemma \ref{peripherian}, $x\in L(T)$.
Let $P := x_0\dots x_r$ be a longest path of $T$ starting from  $x_0=x$. Then $3\le r\le n-2$.
Note that  $\delta_T(x_i)\ge 3$ for some $i$ with  $1\le i\le r-1$.
By the choice of $T$,
$\delta_T(x_i)=2$ for $i=1,\dots, r-2$, as, otherwise, by moving an edge outside $P$  from $x_i$ to $x_{r-1}$ we get a tree $T'$, for which we have $LS(T')\ge  s_{T'}(x,L(T'))>s_{T}(x,L(T))=LS(T)$, which is a contradiction.
As $P$ is a longest path from  $x$, all neighbors of $x_{r-1}$ except $x_{r-2}$ are leaves. Let $a=\delta_T(x_{r-1})-1$. Then $T\cong P_{n,a}$, and
\[
LS(T)=s_T(x,L(T))=a(n-a)
\]
which is maximized to  $\lfloor\frac{n^2}{4}\rfloor$  if and only if $a=\frac{n}{2}$ if $n$ is even, and  $a=\frac{n-1}{2}, \frac{n+1}{2}$ if $n$ is odd.
\end{proof}

\begin{Theorem} Let $T$ be a tree of order $n$ with diameter $d$, where $2\le d\le n-1$.
Let $t=\left\lceil\frac{2(n-1-d)}{d}\right\rceil$ for even $d$ and $t=\left\lceil\frac{2(n-1-d)}{d-1}\right\rceil$ for odd $d$.
Then
\[
LS(T)\ge n-1+\left\lceil\frac{dt}{2}\right\rceil
\]
with equality if and only if $T$ is a tree with a diametric path between two leaves $x$ and $y$ and  exactly $t$ hanging paths at vertices of the diametric path such that $LS(T)=s_T(x,L(T))=s_T(y,L(T))$ for even $dt$ and $LS(T)=\max\{s_T(x,L(T)),  s_T(y,L(T))\}$ and
$|s_T(x,L(T))- s_T(y,L(T))|=1$
for odd $dt$.
\end{Theorem}

\begin{proof} Let $P:= v_0v_1\dots v_d$ be a diametric path in $T$. Let $L^*(T)=L(T)\setminus\{v_0, v_{d}\}$.
Note that
\[
\sum_{w\in L^*(T)}d_T(w,P)\ge |E(T)|-d=n-1-d
\]
with equality if and only if each vertex outside $P$ has degree one or two in $T$.
For convenience,  denote by $t=\lceil\frac{2(n-1-d)}{d}\rceil$ for even $d$ and $t=\lceil\frac{2(n-1-d)}{d-1}\rceil$ for odd $d$. By a result of Qiao and Zhan~\cite{QZ}, $|L(T)|\ge t+2$. So $|L^*(T)|\ge t$.
Then
\begin{eqnarray}\label{eq1}
LS(T)&\ge& \max\{s_T(v_{0},L(T)),s_T(v_{d},L(T))\}\nonumber \\
&\ge&\frac{s_T(v_{0},L(T))+s_T(v_{d},L(T))}{2}\nonumber \\
&=&d+\sum_{w\in L^*(T)}\frac{d_T(v_0,w)+d_T(v_{d},w)}{2} \nonumber\\
&=&d+\frac{|L^*(T)|d}{2}+\sum_{w\in L^*(T)}d_T(w,P)\\
&\ge&d+\frac{dt}{2}+n-1-d \nonumber\\
&=&n-1+\frac{dt}{2}.\nonumber
\end{eqnarray}
So
\[
LS(T)\ge n-1+\left\lceil\frac{dt}{2}\right\rceil.
\]

Suppose  that $LS(T)= n-1+\left\lceil\frac{dt}{2}\right\rceil$. By the proof of \eqref{eq1}, each vertex outside $P$ has degree one or two in $T$.
Note also that  $|L^*(T)|=t$. Otherwise,   $|L^*(T)|\ge t+1$. So, by \eqref{eq1},
\[
n-1+\left\lceil\frac{dt}{2}\right\rceil=LS(T) \ge d+\frac{|L^*(T)|d}{2}+n-1-d\ge  n-1+\frac{d(t+1)}{2},
\]
a contradiction. So $T$ is a tree with a diametric path $P$ and exactly $t$ hanging paths at vertices of $P$.

If $dt$ is even, then the three inequalities in \eqref{eq1} must be equalities, so  $LS(T)=s_T(v_{0},L(T))=s_T(v_{d},L(T))$

Suppose next that $dt$ is odd.  From \eqref{eq1}, we have
\[
\max\{s_T(v_{0},L(T)),s_T(v_{d},L(T))\}-\frac{s_T(v_{0},L(T))+s_T(v_{d},L(T))}{2}\le \frac{1}{2},
\]
i.e., $|s_T(v_{0},L(T))-s_T(v_{d},L(T))|=0,1$. As $|L^*(T)|=t$, we have  $|s_T(v_{0},L(T))-s_T(v_{d},L(T))|=1$ by \eqref{eq1}.

From \eqref{eq1}, we also have
$LS(T)=\max\{s_T(v_{0},L(T)),s_T(v_{d},L(T))\}$, as otherwise,
\begin{eqnarray*}
n-1+\left\lceil\frac{dt}{2}\right\rceil &=& LS(T)\\
&>& \max\{s_T(v_{0},L(T)),s_T(v_{d},L(T))\} \\
&=&\frac{s_T(v_{0},L(T))+s_T(v_{d},L(T))+1}{2}\\
&=&n-1+\left\lceil\frac{dt}{2}\right\rceil,
\end{eqnarray*}
a contradiction.

Conversely, if $T$ is a tree with a diametric path between two leaves $x$ and $y$ and  exactly $t$ hanging paths at vertices of the diametric path such that $LS(T)=s_T(x,L(T))=s_T(y,L(T))$ for even $dt$ and $|s_T(x,L(T))- s_T(y,L(T))|=1$ and $LS(T)=\max\{s_T(x,L(T)),  s_T(y,L(T))\}$ for odd $dt$,
then \begin{eqnarray*}
LS(T)&=&\max\{s_T(x,L(T)),s_T(y,L(T))\}\\ [2mm]
&=&\left\lceil\frac{s_T(x,L(T))+s_T(y,L(T))}{2}\right\rceil\\
&=&n-1+\left\lceil\frac{dt}{2}\right\rceil,
\end{eqnarray*}
as desired.
\end{proof}

We give an example on trees of order $15$ with diameter $8$. The three trees in Fig.~\ref{fig:1} are the ones that minimize the maximum leaf status.

\begin{figure}[htbp]
\centering
\includegraphics[width=12.8cm,height=1.6cm]{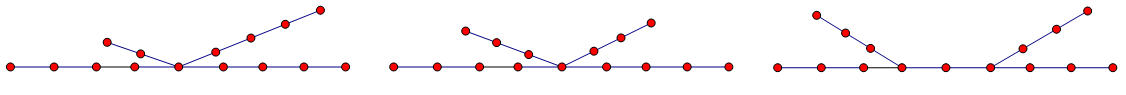}
\caption{Three trees of order $15$ with diameter $8$.}
\label{fig:1}
\end{figure}

\begin{Theorem} Let $T$ be a tree of order $n$ with diameter $d$, where $2\le d\le n-1$. Then
\[
LS(T)\le d(n-d)
\]
with equality if and only if $T\cong P_{n,n-d}$.
\end{Theorem}

\begin{proof} The result is trivial if $d=2, n-1$. Suppose that $3\le d\le n-2$.
Let $T$ be a tree of order $n$ with diameter $d$ that maximizes the maximum leaf status.

Let $x$ be a leaf peripherian vertex of $T$. By Lemma \ref{peripherian}, $x\in L(T)$.
We want to show that $x$ lies on some diametric path of $T$. Suppose that this is not true. That is, $x$ lies outside any diametric path.
 Let $P:= x_0 x_1\dots x_d$ be an arbitrary  diametric path of $T$. Then
$d_T(x, P)=d_T(x,x_i)$ for some $i$ with $2\le i\le d-2$ and $1\le d_T(x,x_i)\le\min\{i-1,d-i-1\}$.
By the choice of $T$, any vertex in the path $Q$ connecting $x$ and $x_i$ except $x$ and $x_i$ (if any exists) has degree two in $T$, as, otherwise, we move an edge outside $Q$ from this vertex to $x_i$ to form a tree $T'$, and for $T'$, its diameter is still $d$, but  $LS(T')\ge s_{T'}(x,L(T'))>s_{T}(x,L(T))=LS(T)$, which is a contradiction.
 Let $r=\delta_T(x_i)$. Denote by $z_1,\dots,z_{r}$ the neighbors of $x_i$ in $T$, where $z_1$ lies on the path $Q$, $z_2=x_{i-1}$ and $z_3=x_{i+1}$. For $i=1,\dots,r$, let $B_i$ be the branch of $T$ at $x_i$ containing $z_{i}$, and  let $a_i=|V(B_i)\cap L(T)|$. Suppose that $r\ge 4$. Let $T'=T-x_iz_1+x_{i-1}z_1$. and $T''=T-x_iz_1+x_{i+1}z_1$.  Note  the diameters of $T'$ and $T''$ are both $d$. By direct calculation,
\begin{eqnarray*}
s_{T'}(x,L(T'))-s_T(x,L(T))&=&\sum_{j=3}^r\sum_{w\in V(B_j)\cap L(T)}1-\sum_{w\in V(B_{2})\cap L(T)}1\\
&=&\sum_{j=3}^{r}a_j-a_2\\
&=&\sum_{j=2}^ra_j-2a_2,
\end{eqnarray*}
and similarly,
\[
s_{T''}(x,L(T''))-s_T(x,L(T))=\sum_{j=2}^ra_j-2a_3.
\]
So $s_{T'}(x,L(T'))-s_T(x,L(T))+s_{T''}(x,L(T''))-s_T(x,L(T))=2\sum_{j=4}^ra_j>0$, implying that
either $LS(T')\ge s_{T'}(x,L(T'))>LS(T)$ or $LS(T'')\ge s_{T''}(x,L(T''))>LS(T)$, a contradiction.
It follows that  $r=3$. Suppose that  there is a leaf  of $T$ different from $x$ such that it is not  adjacent to $x_1$ or $x_{d-1}$. Choose such a leaf $y$ of $T$ such that $d_T(y,P)$ is as large as possible. Assume that $d_T(y, P)=d_T(y,x_j)$, where $2\le j\le d-2$ with $j\ne i$. Assume that $j<i$. Let $y'$ be the unique neighbor of $y$ in $T$. By moving the leaf edges at $y'$  from $y'$ to $x_1$ we get a tree $T'$, for which we have $LS(T')-LS(T)\ge s_{T'}(x, L(T'))-s_{T}(x, L(T))\ge d_T(x,y')>0$, so
$LS(T')>LS(T)$, a contradiction. This shows that $\delta_T(x_k)=2$ for any $2\le k\le d-2$ with $k\ne i$.

Let $\ell=d_T(x, x_i)$.  Assume that $i\le \lfloor\frac{d}{2}\rfloor$. By the choice of $T$, $\delta_T(x_1)=2$ and $\delta_T(x_{d-1})=n-d-\ell$.
So $LS(T)=s_T(x,L(T))=\ell+i+(\ell+d-i)(n-d-\ell-1)$ and $s_T(x_0,L(T))=(n-d-\ell-1)d+\ell+i$.
But  $s_T(x_0,L(T))-LS(T)=(n-d-\ell-1)(i-\ell)>0$, a contradiction. Thus $x$ lies on some diametric path of $T$.

Let $P':=xv_1\dots v_d$ be a diametric path starting from $x$. By the choice of $T$ and the argument as above, all leaves of $T$ different from  $x$ are adjacent to $v_{d-1}$, i.e.,  $T\cong P_{n,n-d}$.  The result follows by noting that $LS(T)=d(n-d)$.
\end{proof}

\section{Minimum internal status}

\begin{Proposition}\label{ismin} Let $T$ be a tree of order $n\ge 3$.
\begin{enumerate}
\item[(i)] $is(T)\ge 0$
with equality if and only if $T\cong S_n$.

\item[(ii)]  If $T\ncong S_n$, then $is(T)\ge 1$ with equality if and only if $T$ is a double star.

\item[(iii)]
 If $T\ncong S_n$, and $T$ is not a double star, then $is(T)\ge 2$ with equality if and only if $T$ is a caterpillar of diameter $4$.
\end{enumerate}
\end{Proposition}

\begin{proof}
Item (i) follows from the fact that  $S_n$ is the only tree with exactly one internal vertex. Item (ii) follows as the double stars are the only trees with exactly two (adjacent) internal vertices. Item (iii) follows as for any tree $T$ of diameter at least $4$ contains three internal vertices inducing  a path $P_3$ in $T$, and if $T$ has more than $3$ internal vertices, then the subtree induced by internal vertices in $T$ contains $S_4$ or $P_4$ so that $is(T)\ge 3$.
\end{proof}

Furthermore, we have

\begin{Proposition} Suppose that $T$ is a tree of order $n$ with diameter $d$, where $2\le d\le n-1$. Then
\[
is(T)\ge \left\lfloor\frac{(d-1)^2}{4}\right\rfloor
\]
with equality if and only if $T$ is caterpillar.
\end{Proposition}

\begin{proof}
By Lemma~\ref{newlemma}, an $I(T)$-centroid vertex is an internal vertex of $T$.

Let $P:= v_0\dots v_d$ be a diametric path in $T$. Obviously, $\{v_1, \dots, v_{d-1}\}$ induces in $T$ a path $P_{d-1}$.
If $T$ is a caterpillar, then $is(T)=s(P_{d-1})$.
Suppose that $T$ is not a caterpillar, i.e., $I(T)\setminus V(P)\neq0$. For any $i=1, \dots, d-1$, we have $s_T(v_i,I(T))\ge s(P_{d-1})+\sum_{w\in I(T)\setminus V(P)}d_T(w,v_i)>s(P_{d-1})$. Suppose that $w\in I(T)\setminus V(P)$. Assume that $d_T(w,v_j)=d_T(w,P)$ for some $2\le j\le d-2$. Then we have
\[
s_T(w,I(T))\ge  s_T(w,I(T)\cap V(P))>s_T(v_j,I(T)\cap V(P))\ge
s(P_{d-1}).
\]
Therefore,  $is(T)\ge s(P_{d-1})=\left\lfloor\frac{(d-1)^2}{4}\right\rfloor$ with equality if and only if $T$ is caterpillar.
\end{proof}

\begin{Theorem}
Let $T$ be a tree of order $n\ge 3$. Then
\[
is(T)\le \left\lfloor\frac{(n-2)^2}{4}\right\rfloor
\]
with equality if and only if $T\cong P_n$.
\end{Theorem}

\begin{proof}  By Lemma~\ref{newlemma}, an $I(T)$-centroid vertex is an internal vertex of $T$, so $is(T)$ equals to the minimum status of $T'$, i.e.,  $is(T)=s(T')$, where $T'$ is the tree obtained from $T$ by deleting all leaves. Let $k=|L(T)|$.  By  Proposition 2.1 in \cite{AH}, we have
\[
is(T)=s(T')\le \left\lfloor \frac{(n-k)^2}{4}\right\rfloor\le \left\lfloor \frac{(n-2)^2}{4}\right\rfloor
\]
with equalities if and only if $k=2$ and $T\cong P_n$.
\end{proof}

\section{Maximum internal status}

 Similarly to Proposition~\ref{ismin}, we have

\begin{Proposition} Let $T$ be a tree of order $n\ge 3$. The following statements are true.
\begin{enumerate}
\item[(i)] $IS(T)\ge 1$
with equality if and only if $T\cong S_n$.

\item[(ii)]  If $T\ncong S_n$, then $IS(T)\ge 3$ with equality if and only if $T$ is a double star.

\item[(iii)]
 If $T\ncong S_n$, and $T$ is not a double star, then $IS(T)\ge 6$ with equality if and only if $T$ is a caterpillar of diameter $4$.
\end{enumerate}
\end{Proposition}

Furthermore, we have

\begin{Proposition} Suppose that $T$ is a tree of order $n$ with diameter $d$, where $2\le d\le n-1$. Then
$IS(T)\ge \frac{d(d-1)}{2}$ with equality if and only if $T$ is caterpillar.
\end{Proposition}

\begin{proof}  Let $P:= v_0v_1\dots v_d$ be a diametric   path  in  $T$. Let $I^*(T)=I(T)\setminus\{v_1, \dots, v_{d-1}\}$. Then
\begin{eqnarray*}
IS(T)&\ge& \max\{s_T(v_{0},I(T)),s_T(v_{d},I(T))\}\\[2mm]
&\ge&\frac{s_T(v_{0},I(T))+s_T(v_{d},I(T))}{2}\\
&=&\frac{2\sum_{i=1}^{d-1}i+\sum_{w\in I^*(T)}(d_T(v_0,w)+d_T(v_{d},w))}{2}\\
&\ge&\frac{d(d-1)}{2}
\end{eqnarray*}
with equality if and only if $I^*(T)=\emptyset$, i.e., $T$ is a caterpillar.
\end{proof}

Let $T$ be a tree with $u\in V(T)$. For positive integer $p$, we denote by $T_{u;p}$ the tree consisting of $T$ and a path $P:=uu_1\dots u_p$ such that $u$ is the only common vertex of $T$ and the path $P$.
In this case, we also say that $P$ is `hanging' a path of length $p$ at $u$ in $T_{u;p}$, though
it is really a hanging path of length $p$ at $u$ in  $T_{u;p}$ only when $\delta_{T}(u)\ge 2$.

Let $G_{u; 0}=G$. For nonnegative integer $p$ and $q$, let $G_{u;p,q}=(G_{u;p})_{u; q}$.

\begin{Lemma}\label{pq}
Let $T$ be a nontrivial tree with $u\in V(T)$ Let $p$ and $q$ be positive integers with $p\ge q$. Then $IS(T_{u;p+1,q-1})=IS(T_{u;p,q})$ if $q\ge 2$ and $T$ is a star with center $u$, otherwise $IS(T_{u;p+1,q-1})>IS(T_{u;p,q})$.
\end{Lemma}

\begin{proof}  Let $H=T_{u;p,q}$ and $H'=T_{u;p+1,q-1}$. Let $u_0u_1\dots u_p$ and $v_0v_1\dots v_q$ be the two `hanging'  paths at $u$ in $H$, where $u_0=v_0=u$. Let $x$ be a internal peripherian vertex of $H$.
Let $L^*(T)=L(T)$ if $u\not\in L(T)$ and $L^*(T)=L(T)\setminus\{u\}$ otherwise.
Then $x\in L(H)= L^*(T)\cup \{u_p, v_q\}$ by Lemma~\ref{inperipherian}. As $p\ge q$, we have $s_H(u_p,I(H))\ge s_H(v_q,I(H))$. So we may assume that $x\in L^*(T)\cup\{u_p\}$.
Let $I^*(T)=I(T)$ if $u\in I(T)$ and $I^*(T)=I(T)\cup \{u\}$ otherwise.

\noindent {\bf Case 1.} $x\in L^*(T)$.

If $q\ge 2$, then
\begin{eqnarray*}
s_H(x,I(H))
&=&\sum_{w\in I^*(T)}d_{H}(x,w)+\sum_{i=1}^{p-1}(d_{H}(x,u)+i)+\sum_{i=1}^{q-1}(d_{H}(x,u)+i)\\
&=&\sum_{w\in I^*(T)}d_{H}(x,w)+d_{H}(x,u)(p+q-2)+\sum_{i=1}^{p-1}i+\sum_{i=1}^{q-1}i,
\end{eqnarray*}
and
\begin{eqnarray*}
s_{H'}(x,I(H'))
&=&\sum_{w\in I^*(T)}d_{H'}(x,w)+\sum_{i=1}^{p}(d_{H'}(x,u)+i)+\sum_{i=1}^{q-2}(d_{H'}(x,u)+i)\\
&=&\sum_{w\in I^*(T)}d_{H'}(x,w)+d_{H'}(x,u)(p+q-2)+\sum_{i=1}^{p}i+\sum_{i=1}^{q-2}i.
\end{eqnarray*}
Then $s_{H'}(x,I(H'))-s_H(x,I(H))=p-q+1>0$ and thus
$IS(H')\ge s_{H'}(x,I(H'))>s_H(x,I(H))=IS(H)$.

If $q=1$, then
\begin{eqnarray*}
s_H(x,I(H))
&=&\sum_{w\in I^*(T)}d_{H}(x,w)+\sum_{i=1}^{p-1}(d_{H}(x,u)+i)\\
&=&\sum_{w\in I^*(T)}d_{H}(x,w)+d_{H}(x,u)(p-1)+\sum_{i=1}^{p-1}i,
\end{eqnarray*}
and
\begin{eqnarray*}
s_{H'}(x,I(H'))
&=&\sum_{w\in I^*(T)}d_{H'}(x,w)+\sum_{i=1}^{p}(d_{H'}(x,u)+i)\\
&=&\sum_{w\in I^*(T)}d_{H'}(x,w)+d_{H'}(x,u)p+\sum_{i=1}^{p}i.
\end{eqnarray*}
Then $s_{H'}(x,I(H'))-s_H(x,I(H))=d_{H}(x,u)+p>0$ and thus
$IS(H')\ge s_{H'}(x,I(H'))>s_H(x,I(H))=IS(H)$.

\noindent {\bf Case 2.} $x=u_p$.

Let $q\ge 2$. Then
\begin{eqnarray*}
s_H(x,I(H))
&=&\sum_{w\in I^*(T)}(p+d_{H}(u,w))+\sum_{i=1}^{p-1}i+\sum_{i=1}^{q-1}(p+i)\\
&=&\sum_{w\in I^*(T)}(p+d_{H}(u,w))+p(q-1)+\sum_{i=1}^{p-1}i+\sum_{i=1}^{q-1}i,
\end{eqnarray*}
and
\begin{eqnarray*}
s_{H'}(x,I(H'))
&=&\sum_{w\in I^*(T)}(p+1+d_{H'}(u,w))+\sum_{i=1}^{p}i+\sum_{i=1}^{q-2}(p+1+i)\\
&=&\sum_{w\in I^*(T)}(p+1+d_{H'}(u,w))+(p+1)(q-2)+\sum_{i=1}^{p}i+\sum_{i=1}^{q-2}i.
\end{eqnarray*}
Then $s_{H'}(x,I(H'))-s_H(x,I(H))=|I^*(T)|-1$. If $|I^*(T)|\ge 2$, then $IS(H')\ge s_{H'}(x,I(H'))>s_{H}(x,I(H))=IS(H)$. If  $|I^*(T)|=1$, then  $IS(H')=s_{H'}(x,I(H'))=s_{H}(x,I(H))=IS(H)$.

Let $q=1$. Then
\begin{eqnarray*}
s_H(x,I(H))=\sum_{w\in I^*(T)}(p+d_{H}(u,w))+\sum_{i=1}^{p-1}i,
\end{eqnarray*}
and
\begin{eqnarray*}
s_{H'}(x,I(H'))=\sum_{w\in I^*(T)}(p+1+d_{H'}(u,w))+\sum_{i=1}^{p}i.
\end{eqnarray*}
Then $s_{H'}(x,I(H'))-s_H(x,I(H))=|I^*(T)|+p>0$, and so $IS(H')\ge s_{H'}(x,I(H'))>s_{H}(x,I(H))=IS(H)$.

The result follows by combing the above two cases and noting that $|I^*(T)|=1$ if and only if $T$ is a star with center $u$.
\end{proof}

\begin{Theorem}
Let $T$ be a tree of order $n\ge 3$. Then
\[
IS(T)\le \frac{n^2-3n+2}{2}
\]
with equality if and only if $T\cong P_n$.
\end{Theorem}

\begin{proof} Let $T$ be a tree of order $n$  that maximizes the maximum internal status. Suppose that $T$ is not a path.  Let $x\in L(T)$.  Then we choose a vertex of degree at least three, say $u$, such that $d_T(x,u)$ is as large as possible. Then there are two hanging paths $P$ and $Q$  at $u$ in $T$. By Lemma~\ref{pq}, we can obtain a tree  $T'$ so that $IS(T')>IS(T)$, a contradiction. Thus  $T\cong P_n$.
Evidently, $IS(P_n)=\sum_{i=1}^{n-2}i=\frac{n^2-3n+2}{2}$.
\end{proof}

\begin{Theorem} \label{De} Let $T$ be a tree of order $n$ with maximum degree $\Delta$, where $2\le \Delta\le n-1$. Then $IS(T)\le \frac{1}{2}(n-\Delta)(n-\Delta+1)$ with equality if and only if $T$ is a starlike tree with at least $\Delta-2$ hanging paths being of length one.
\end{Theorem}
\begin{proof}  It is trivial if $\Delta=2$. Suppose that $\Delta\ge 3$.
Let $T$ be a tree of order $n$ with maximum degree $\Delta$ that maximizes the maximum internal status. Let $u$ be a vertex of degree $\Delta$. If there is a vertex different from $u$ with degree at least three, then we may choose such a vertex $v$ by requiring that $d_T(u,v)$ is as large as possible. This implies that there are two hanging paths at $v$ in $T$. By Lemma~\ref{pq}, there is a tree of order $n$ with maximum degree $\Delta$ having larger maximum internal status, which is a contradiction. That is, $u$ is the only vertex of degree at least three. In other words, $T$ is a starlike tree. By Lemma~\ref{pq} again, $\Delta-2$ hanging paths are of length one. Thus $T$ is a starlike tree with at least $\Delta-2$ hanging paths being of length one. The result follows by noting that $IS(T)=IS(P_{n, \Delta-1})=\sum_{i=1}^{n-\Delta}i=\frac{1}{2}(n-\Delta)(n-\Delta+1)$.
\end{proof}

\vspace{5mm}

\noindent {\it Acknowledgement.}
This work was supported by  National Natural Science Foundation of China (No.~11671156).

%
\end{document}